\documentclass[12pt]{article}

\usepackage{amsmath, amsfonts, amssymb, amsthm}
\usepackage{graphicx}
\usepackage{cite}
\usepackage{hyperref}
\usepackage{algorithm}
\usepackage[noend]{algorithmic}
\usepackage{fullpage}
\usepackage{color}
\usepackage{enumerate}

\newcommand{\NP}{\textsf{NP}}
\DeclareMathOperator{\E}{E}

\newtheorem{lemma}{Lemma}
\newtheorem{theorem}{Theorem}

\begin{document}
\title{Randomized Speedup of\\ the Bellman--Ford Algorithm}

\author{Michael J. Bannister and David Eppstein\\
Computer Science Department, University of California, Irvine}

\maketitle

\begin{abstract}
We describe a variant of the Bellman--Ford algorithm for single-source shortest paths in graphs with negative edges but no negative cycles that randomly permutes the vertices and uses this randomized order to process the vertices within each pass of the algorithm. The modification reduces the worst-case expected number of relaxation steps of the algorithm, compared to the previously-best variant by Yen (1970), by a factor of $2/3$ with high probability. We also use our high probability bound to add negative cycle detection to the randomized algorithm.
\end{abstract}

\section{Introduction}

The \emph{Bellman--Ford} algorithm~\cite{Bel-QAM-58,ForFul-62,Moore-ISST} is a label-correcting algorithm for the single-source shortest path problem in directed graphs that may have negatively-weighted edges, but no negative cycles. The algorithm can also be modified to detect negative cycles, when they exist. For a graph with $n$ vertices and $m$ edges, it takes $O(mn)$ time; despite its longevity, this remains the best strongly-polynomial time bound known for this version of the shortest path problem~\cite{AMO-93-Chap5}. If the graph has small integer edge weights, then some newer algorithms whose runtime depends on bounds of the edge weight may be faster~\cite{Shrt-Surv}.

\begin{figure}[b]
\centering
\includegraphics[width=0.475\textwidth]{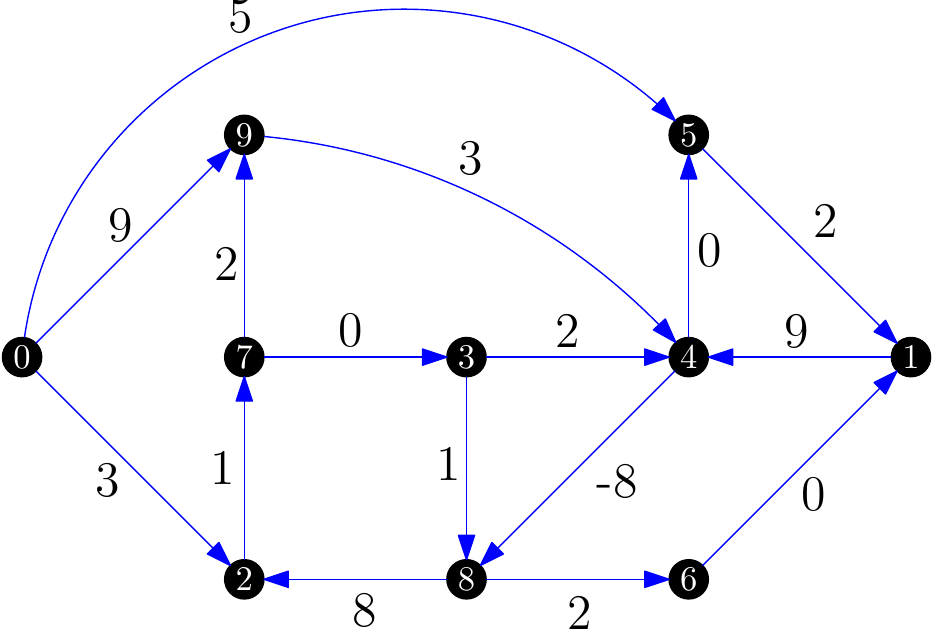}\rule{0em}{0em}\hfill
\includegraphics[width=0.475\textwidth]{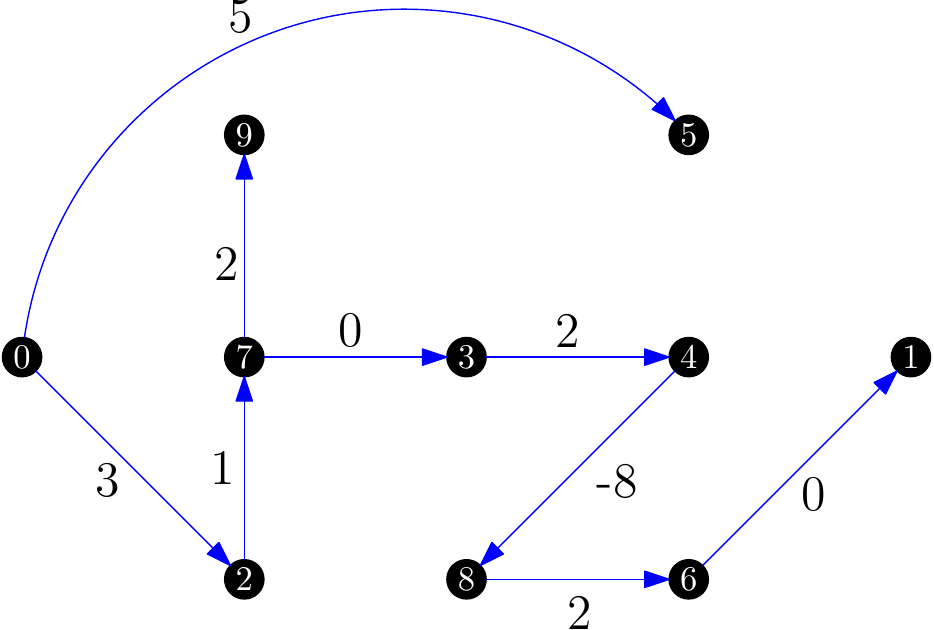}
\caption{Example of a general weighted DAG (left) with its shortest path tree (right).}
\label{fig:example}
\end{figure}

In the absence of improvements to the asymptotic complexity of the algorithm, it becomes of interest to optimize the constant factors in its running time. The bulk of the algorithm's time is spent in \emph{relaxation} or \emph{label-correction} steps in which a candidate value for the distance to a vertex is replaced by the minimum of its previous value and another number. In its most basic form, the algorithm performs at most $mn$ of these relaxation steps, but this can be improved in two ways, both due to Yen. Processing the vertices in a first-in-first-out order that avoids reprocessing vertices whose candidate distance has not changed in the previous step reduces the number of relaxation steps to less that $n^3/2$, an improvement by a factor of two for dense graphs~\cite{Yen-Book,Dreyfus-Survey,Lawler-Yen}. The second improvement, published in 1970 by Yen~\cite{Yen-QAM-70} and since repeated in several textbooks~\cite{CLRS-Yen,HuShi-Yen,Lawler-Yen,PemSki-03, Chen}, involves partitioning the input directed graph into two directed acyclic graphs and alternating between passes of the algorithm that relax the edges in one of these two DAGs. This method reduces the number of relaxation steps to $mn/2+m$, an improvement by nearly a factor of two over the original algorithm even for sparse graphs. Both improvements can be combined to yield less than $n^3/4$ relaxation steps for dense graphs~\cite{Lawler-Yen,Yen-Book, Dreyfus-Survey}.

In this paper we combine these previous ideas with an additional idea: randomly permuting the vertices of the graph, and using that permutation to order the vertices within each pass of the algorithm. As we show, this modification produces another factor of $2/3$ speedup in the expected time for the algorithm on a worst-case input. With this modification, the algorithm performs at most $mn/3+m$ relaxation steps in expectation, or $mn/3+o(mn)$ steps with high probability. For dense graphs, it takes at most $n^3/6$ steps in expectation and $n^3/6+o(n^3)$ steps with high probability. Finally, we use the high probability bounds to detect the presence of negative cycles in the input graph in $mn/3 + o(mn)$ time with high probability.

Despite the simplicity of the method, we thus obtain a large constant-factor savings in runtime. Additionally, this improvement makes an interesting test case for randomization in basic graph algorithms. Indeed, after the appearance of our initial blog posting at \url{http://11011110.livejournal.com/215330.html} describing the simplest version of this result (the expected analysis of the sparse case), the same result has also been used as a web exercise for Sedgewick's {\it Algorithms} textbook \cite{Sedgewick}.
\section{Previous Algorithms}

The Bellman--Ford algorithm is an instance of a class of algorithms known as \emph{relaxation algorithms} or \emph{label-correction algorithms} for finding shortest paths from a designated start vertex $s$ to all other vertices in a given directed graph. These algorithms maintain for each vertex $v$ a tentative distance $D[v]$ and a tentative predecessor $P[v]$, with the invariant that the tentative distance $D[v]$ is always an upper bound on the true distance $d(s,v)$ from $s$ to~$v$. Initially, $D[s]=0$ and $D[v]=+\infty$ for every $v\ne s$; $P[v]$ is undefined. Then, the algorithm performs a sequence of relaxation steps in which it calls the \texttt{relax} procedure described in Algorithm~\ref{alg:relax}.

\begin{algorithm}
\caption{Procedure \texttt{relax}$(u,v)$: relax the edge from $u$ to $v$.}
\label{alg:relax}
\begin{algorithmic}
\IF {$D[v] > D[u]+\textrm{length}(u,v)$}
  \STATE $D[v] \gets D[u]+\textrm{length}(u,v)$
  \STATE $P[v]\gets u$
\ENDIF
\end{algorithmic}
\end{algorithm}

We say that $v$ is \emph{accurate} if $D[v]$ holds the correct distance from $s$ to $v$; initially, $s$ itself is accurate and all other vertices are not. We define a \emph{correct relaxation} to be a call to \texttt{relax}$(u,v)$ for an edge from $u$ to $v$ that belongs to a shortest path from $s$ to $v$, at a time when $u$ is accurate and $v$ is not accurate. After a correct relaxation, $v$ will become accurate. The Bellman--Ford algorithm is based on the insight that, if we relax all of the edges in the graph, then at least one correct relaxation is guaranteed to occur. After $n-1$ correct relaxations, all distances must be correct. Once this happens, each $P[v]$ points to the predecessor of $v$ on a valid shortest path from $s$ to $v$.

\begin{algorithm}
\caption{The basic Bellman--Ford algorithm}
\label{alg:bellford}
\begin{algorithmic}
\FOR {$i\gets 1$ to $n-1$}
  \FOR {each edge $uv$ in graph $G$}
    \STATE \texttt{relax}$(u,v)$
  \ENDFOR
\ENDFOR
\end{algorithmic}
\end{algorithm}

This version of the algorithm performs $m(n-1)$ calls to the relax procedure, regardless of the input. A simple optimization is possible: only relax edges from vertices $u$ for which $D[u]$ has recently changed, since other vertices cannot lead to correct relaxations. Additionally, the algorithm may be terminated early when no recent changes exist. For sparse graphs, this may be a practical improvement but does not change the worst case running time significantly. However, for dense graphs the improvement is larger: after the $i$th iteration of the outer loop of the algorithm, $i+1$ vertices will already have their correct distances and will no longer change, so in the $i$th iteration at most $n-i$ vertices can have recently changed, and the number of relaxations within that iteration is at most $(n-1)(n-i)$. Adding this up over all iterations (and using the observation that in the first iteration of the outer loop we need only relax the edges that go out of $s$) produces a total of $(n-1)((n-1)(n-2)/2+1)< n^3/2$ relaxations, a significant improvement over the basic Bellman--Ford algorithm for dense graphs.

\begin{algorithm}
\caption{Adaptive Bellman--Ford with early termination}
\label{alg:adaptive}
\begin{algorithmic}
\STATE $C\gets\{s\}$
\WHILE {$C\ne\emptyset$}
  \FOR {each vertex $u$ in $C$}
    \FOR {each edge $uv$ in graph $G$}
      \STATE \texttt{relax}$(u,v)$
    \ENDFOR
  \ENDFOR
  \STATE $C\gets\{$vertices $v$ for which $D[v]$ changed$\}$
\ENDWHILE
\end{algorithmic}
\end{algorithm}

As Yen~\cite{Yen-QAM-70} observed, it is also possible to improve the algorithm in a different way, by choosing more carefully the order in which to relax the edges within each iteration of the outer loop so that two correct relaxations can be guaranteed for each iteration except possibly the last. Specifically, number the vertices arbitrarily starting from the source vertex, let $G^{+}$ be the subgraph formed by the edges that go from a lower numbered vertex to a higher numbered vertex, and let $G^{-}$ be the subgraph formed by the edges that go from a higher numbered vertex to a lower numbered vertex. Then $G^{+}$ and $G^{-}$ are both directed acyclic graphs, and the numbering of the vertices is a topological numbering of $G^{+}$ and the reverse of a topological numbering for $G^{-}$. Each iteration of Yen's algorithm processes each of these two subgraphs in topological order.

\begin{algorithm}
\caption{Yen's algorithm (adaptive version with early termination)}
\label{alg:yen}
\begin{algorithmic}
\STATE number the vertices arbitrarily, starting with $s$
\STATE $C\gets\{s\}$
\WHILE {$C\ne\emptyset$}
  \FOR {each vertex $u$ in numerical order}
    \IF {$u\in C$ or $D[u]$ has changed since start of iteration}
      \FOR {each edge $uv$ in graph $G^{+}$}
        \STATE \texttt{relax}$(u,v)$
      \ENDFOR
    \ENDIF
  \ENDFOR
  \FOR {each vertex $u$ in reverse numerical order}
    \IF {$u\in C$ or $D[u]$ has changed since start of iteration}
      \FOR {each edge $uv$ in graph $G^{-}$}
        \STATE \texttt{relax}$(u,v)$
      \ENDFOR
    \ENDIF
  \ENDFOR
  \STATE $C\gets\{$vertices $v$ for which $D[v]$ changed$\}$
\ENDWHILE
\end{algorithmic}
\end{algorithm}

Suppose that, at the start of an iteration of the outer loop of the algorithm, vertex $u$ is accurate, and that $\pi$ is a path in the shortest path tree rooted at $s$ that starts at $u$, with all vertices of $\pi$ inaccurate. Suppose also that  all of the edges in $G^{+}$ that belong to path $\pi$ are earlier in the path than all of the edges in $G^{-}$. Then, in that single iteration, the steps that relax the edges of $G^{+}$ in a topological ordering of $G^{+}$ will correctly relax all of the edges in $\pi\cap G^{+}$, and then the steps that relax the edges of $G^{-}$ in a topological ordering of $G^{-}$ will correctly relax all of the edges in $\pi\cap G^{-}$. Therefore, after the iteration, all vertices in $\pi$ will be accurate. More generally, if $k$ is the maximum number of times that any shortest path of the given graph alternates between edges in $G^{+}$ and $G^{-}$, then after $k$ iterations of the algorithm every vertex will be accurate and after $k+1$ iterations the algorithm will terminate. Therefore, the algorithm will perform at most $km+m$ relaxation steps. For any graph, $k\le n/2$, so the algorithm performs a total of at most $mn/2+m$ relaxation steps in the worst case.

A similar analysis applies also to dense graphs. With the possible exception of the final iteration, each iteration of Yen's algorithm increases the number of accurate vertices by at least two; once a vertex becomes accurate, it can be the first argument of a relaxation operation in only a single additional iteration of the algorithm. Therefore, iteration $i$ relaxes at most $n(n-2i)$ edges. Summing over all iterations yields a total number of relaxations that is less than $n^3/4$. Experiments conducted by Yen have demonstrated the practicality of these speedups in spite of extra time needed to maintain the set of recently changed vertices~\cite{Yen-Book}.

\begin{figure}[t]
\centering
\includegraphics[width=0.475\textwidth]{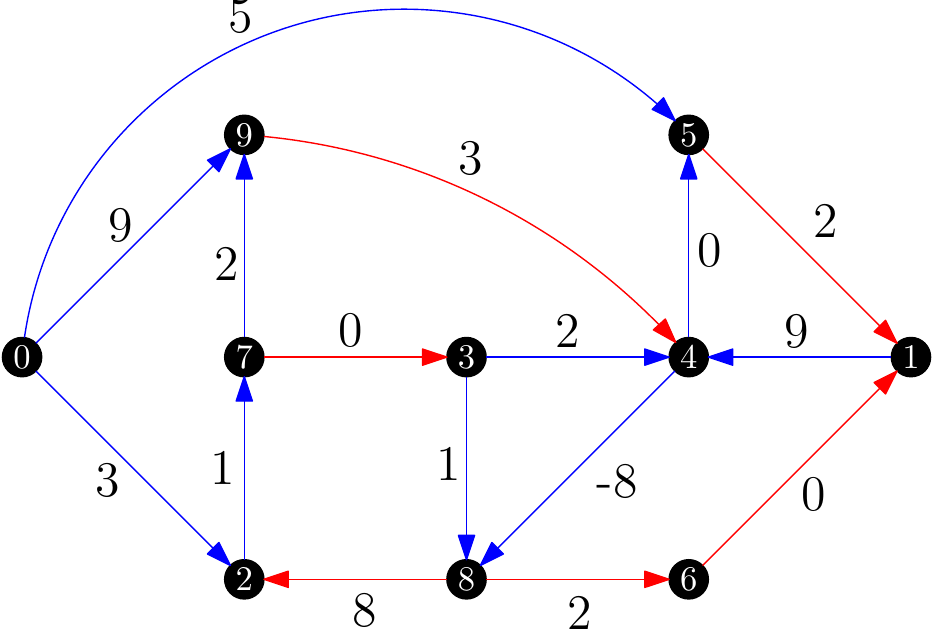}\rule{0em}{0em}\hfill
\includegraphics[width=0.475\textwidth]{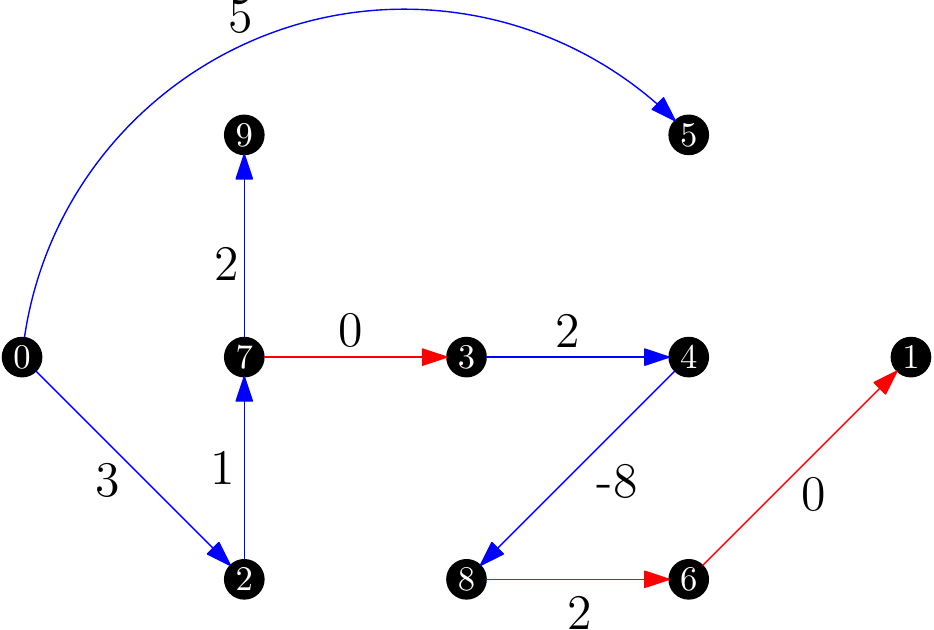}
\caption{Example from Figure \ref{fig:example} with edges in $G^-$ colored red.}
\label{fig:example2}
\end{figure}

\section{The Randomized Algorithm}

Our randomized algorithm makes only a very small change to Yen's algorithm, by choosing the numbering of the vertices randomly rather than arbitrarily. In this way, it makes the worst case of Yen's algorithm (in which a shortest path alternates between edges in $G^+$ and $G^-$) very unlikely.

\begin{algorithm}
\caption{Randomized variant of the Bellman--Ford algorithm}
\label{alg:rand}
\begin{algorithmic}
\STATE number the vertices randomly such that all permutations with $s$ first are equally likely
\STATE $C\gets\{s\}$
\WHILE {$C\ne\emptyset$}
  \FOR {each vertex $u$ in numerical order}
    \IF {$u\in C$ or $D[v]$ has changed since start of iteration}
      \FOR {each edge $uv$ in graph $G^{+}$}
        \STATE \texttt{relax}$(u,v)$
      \ENDFOR
    \ENDIF
  \ENDFOR
  \FOR {each vertex $u$ in reverse numerical order}
    \IF {$u\in C$ or $D[v]$ has changed since start of iteration}
      \FOR {each edge $uv$ in graph $G^{-}$}
        \STATE \texttt{relax}$(u,v)$
      \ENDFOR
    \ENDIF
  \ENDFOR
  \STATE $C\gets\{$vertices $v$ for which $D[v]$ changed$\}$
\ENDWHILE
\end{algorithmic}
\end{algorithm}

To analyze the algorithm, we first consider the structure of its worst-case instances.

\begin{lemma}
\label{lem:alg-is-combinatorial}
Let $G$ and $s$ define an input to Algorithm~\ref{alg:rand}. Then the number of iterations of the outer loop of the algorithm depends only on the combinatorial structure of the subgraph $S$ of $G$ formed by the set of edges belonging to shortest paths of $G$; it does not depend in any other way on the weights of the edges in $G$.
\end{lemma}

\begin{proof}
In each iteration, a vertex $v$ becomes accurate if there is a path $\pi$ in $S$ from an accurate vertex $u$ to $v$ with the property that $\pi$ is the concatenation of a path $\pi^+\in S\cap G^+$ with a path $\pi_-\in S\cap G^-$. This property does not depend on the edge weights.
\end{proof}

\begin{lemma}
\label{lem:worst-case-struc}
Among graphs with $n$ vertices and $m$ edges, the worst case for the number of iterations of Algorithm~\ref{alg:rand} is provided by a graph in which there is a unique shortest path tree in the form of a single $(n-1)$-edge path.
\end{lemma}

\begin{proof}
Let $G$ and $s$ be an input instance for the algorithm,
and as above let $S$ be the set of edges that belong to shortest paths from $s$ in $G$. If $S$ contains two edges into a vertex $v$, then increasing the weight of one of them (causing it to be removed from $S$) can only reduce the sets of vertices that become accurate in each iteration of the algorithm, as described in Lemma~\ref{lem:alg-is-combinatorial}. Thus, the modified graph has at least as large an expected number of iterations as $G$. Similarly, if there are two edges $vu$ and $vw$ exiting vertex $v$, then replacing edge $vw$ by an edge $uw$ whose weight is the difference of the two previous edges leaves the distance to $w$ unchanged (and therefore does not change any of the rest of $S$) while increasing the number of steps from $s$ to $w$ and its descendants; again, the expected number of iterations in the modified graph is at least as large as it was prior to the modification. By repeating such modifications until no more can be performed, the result is a graph in the form given by the statement of the lemma.
\end{proof}

For the tail bounds on the runtime we will use the methods of bounded differences which is restated in Lemma \ref{lem:bounded-diff}.

\begin{lemma}[Method of Bounded Differences \cite{McDiarmid, Dubhashi-Panconesi}]
\label{lem:bounded-diff}
If $f$ is Lipschitz (w.r.t. Hamming distance) with constants $d_i$, for $1 \leq i \leq n$, and $X_i$ are independent random variables, then
\[
\Pr[f > \E[f] + t] \leq \exp \left( - \frac{2t^2}{d} \right) \quad\text{and}\quad
\Pr[f > \E[f] - t] \leq \exp \left( - \frac{2t^2}{d} \right)
\]
where $d = \sum d_i^2$.
\end{lemma}

From our previous analysis of Yen's algorithm we see that each iteration processes the vertices on a shortest path up to the first local minimum in the sequence of vertex labels. For this reason we will be interested in the distribution of local minima in random sequences. The problem of counting local minima is closely related to the problem of determining the length of the longest alternating subsequence \cite{ALT1, ALT2}.

\begin{figure}[b]
\centering
\includegraphics[width=0.475\textwidth]{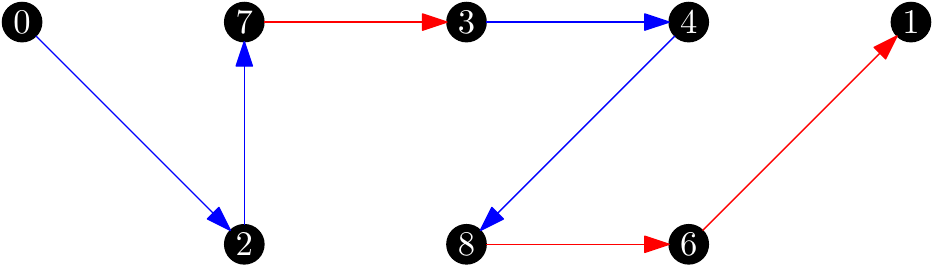}\rule{0em}{0em}\hfill
\includegraphics[width=0.475\textwidth]{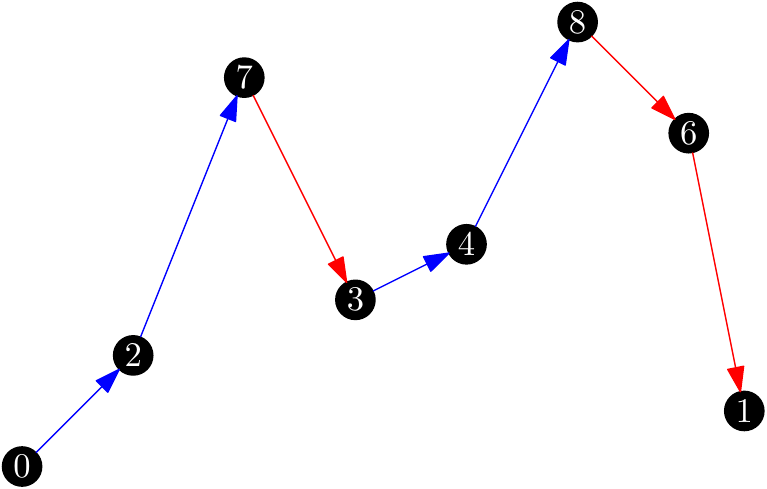}
\caption{Longest shortest path from Figure \ref{fig:example} (left) with height used to represent vertex label (right).}
\label{fig:example3}
\end{figure}

\begin{lemma}
\label{lem:sequence}
If $X_1, \ldots, X_n$ is a sequence of random variables for which ties have probability zero and each permutation is equally likely (e.g. i.i.d. real random variables), then
\begin{enumerate}[\quad(1)]
\item the expected number of local minima is $(n-2)/3$ not counting endpoints;
\item and, the probability that there are more than 
\[
\frac{n-2}{3} + \sqrt{2cn \log n} \leq \frac{n-2}{3}\left(1 + 3\sqrt{2} \sqrt{\frac{c\log n}{n}}\right)
\] local minima is at most $1/n^c$.
\end{enumerate}
\end{lemma}
\begin{proof}
For (1) notice that there are six ways that $X_{j-1}, X_j, X_{j+1}$ may be ordered when $1 < j<n$, and two of these orderings make $X_j$ a local minima. For (2) let $f(X_1, \ldots, X_n)$ equal the number of local minima in the sequence. Changing any one of the $X_i$ changes the value of $f(X_1, \ldots, X_n)$ by at most $2$. Hence by Lemma~\ref{lem:bounded-diff} with $t = \sqrt{2cn\log n}$ the statement in $(2)$ holds. 
\end{proof}

\begin{theorem}
\label{thm:run-time}
The expected number of relaxations performed by Algorithm~\ref{alg:rand} (on a graph with at least three vertices) is at most $mn/3+m$, and the number of relaxations is less than
\[
\frac{mn}{3} + m +m\sqrt{2cn\log n} \leq \left( \frac{mn}{3} + m \right) \left(1 + 3\sqrt{2}\sqrt{\frac{c\log n}{ n}}\right)
\]
with probability at least $1 - 1/n^c$.
\end{theorem}

\begin{proof}
Let $G$ be a worst-case instance of the algorithm, as given by Lemma~\ref{lem:worst-case-struc}. In each iteration of the algorithm other than the first and last, let $v$ be the last accurate vertex on the single maximal shortest path in $G$. Since this is neither the first nor the last iteration, $v$ must be neither the first nor the last vertex on the path; let $u$ be its predecessor and let $w$ be its successor.  Then, in order for $v$ to have become accurate in the previous iteration without letting $w$ become accurate as well, it must be the case that $v$ is the first of the three vertices $\{u,v,w\}$ in the ordering given by the random permutation selected by the algorithm: if $u$ were first then edge $uv$ would belong to $G^+$ and no matter whether edge $vw$ belonged to $G^+$ or $G^-$ it would be relaxed later than $uv$ in the same iteration. And if $w$ were first then $vw$ would belong to $G^-$ and would be relaxed later than $uv$ in each iteration no matter whether $uv$ belonged to $G^+$ or $G^-$.

Thus, we may bound the expected number of iterations of Algorithm~\ref{alg:rand} on this input by bounding the number of vertices $v$ that occur earlier in the random permutation than both their predecessor and their successor in the shortest path, i.e., the local minima in sequence of labels. The start vertex $s$ is already assumed accurate so applying Lemma~\ref{lem:sequence} to the remaining $n-1$ vertices yields $(n-3)/3$ iterations for the interior vertices. Therefore, the expected number of iterations is $2+(n-3)/3=(n+3)/3$. Each iteration relaxes at most $m$ edges, so the total expected number of relaxations is at most $mn/3+m$. An application of the second part of Lemma~\ref{lem:sequence} finishes the proof.
\end{proof}

Lemma~\ref{lem:worst-case-struc} does not directly apply to the dense case, because we need to bound the number of relaxations within each iteration and not just the number of iterations. Nevertheless the same reasoning shows that the same graph (a graph with a unique shortest path tree in the form of a single path) forms the worst case of the algorithm.

\begin{theorem}
For dense graphs the expected number of relaxations performed by Algorithm~\ref{alg:rand} is at most $n^3/6$, and the number of relaxations is less than
\[
\frac{n^3}{6} + \sqrt{2} n^{5/2} \sqrt{c\log n} \leq \frac{n^3}{6} \left(1+ \sqrt{2}\sqrt{\frac{c\log n}{n}}\right)
\]
with probability $1 - 1/n^{c-1}$.
\end{theorem}
\begin{proof}
Let $v$ be a vertex in the input graph whose path from $s$ in the shortest path tree is of length $k$.
Then the expected number of iterations needed to correct $v$ is $k/3$, assuming the worst case that $v$ is processed in each of these iterations we will relax at most
\[
n\sum_{k=1}^n \frac{k}{3} \leq n^3 / 6
\] 
edges. Also, Theorem~\ref{thm:run-time} implies that $v$ will be corrected after at most
\[
k/3 + \sqrt{2ck\log n}
\]
with probability at least $1-1/n^c$. Again, assuming the worst case, the edges from $v$ will be relaxed in each iteration, we will relax at most
\[
n\sum_{k=1}^n k/3 + \sqrt{2ck\log k} \leq n^3/6 + \sqrt{2c} n^{5/2} \sqrt{\log n}
\]
edges with probability at least $1-1/n^{c-1}$.
\end{proof}

\section{Negative Cycle Detection}

If $G$ is a directed-acyclic graph with a negative cycle reachable from the source, then the distance to some vertices is effectively $-\infty$. If we insist on finding shortest simple paths, then the problem is \NP-hard\cite{Garey-Johnson}.

Because of this difficulty, rather than seeking the shortest simple paths we settle for a timely notification of the existence of a negative cycle. There are several ways in which single-source-shortest-path algorithms can be modified to detect the presence of negative cycles~\cite{Neg-Surv}. We will use what is commonly referred to as subtree traversal. After some number of iterations of the Bellman--Ford algorithm, define $G_p$ to be the parent graph of $G$; this is a graph with the same vertex set as $G$ and with an edge from $v$ to $u$ whenever the tentative distance $D[v]$ was set by relaxing the edge in $G$ from $u$ to $v$. That is, for each $v$ other than the start vertex, there is an edge from $v$ to $P[v]$. Cycles in $G_p$ correspond to negative cycles in $G$~\cite{Tarjan-DS}. Moreover, if $G$ contains a negative cycle, then after $n-1$ iterations $G_p$ will contain a cycle after each additional iteration~\cite{Neg-Surv}. We would like to lower this requirement from $n-1$ to something more in line with runtime of Algorithm~\ref{alg:rand}.

For each vertex $v$ in any input graph $G$ there exists a shortest simple path from the source $s$ to $v$; denote the length of this path by $D'[v]$. This quantity $D'[v]$ will not be calculated by our algorithm, but we will use it in our analysis. If $G$ has a negative cycle, then at some point it will be the case that $D[v] < D'[v]$ for at least one vertex $v$ in $G$.

\begin{lemma}[Cf. \cite{Neg-Surv}]\label{lem:dist2cycle}
If after an iteration of Algorithm~\ref{alg:rand} we have $D[v] < D'[v]$ for some vertex $v$, then $G_p$ has a cycle.
\end{lemma}

\begin{lemma}\label{lem:all-ssp}
After $n/3 + 1 + \sqrt{2cn\log n}$ iterations $D[v] \leq D'[v]$ for all $v$ with probability at least $1 - 1/n^{c-1}$.
\end{lemma}
\begin{proof}
Let $v$ be a vertex in the input graph, and $u_0 = s, u_1, \ldots, u_n=v$ the shortest simple path to $v$ from the source $s$. Then the proof of Theorem~\ref{thm:run-time} shows that the edges $u_0u_1, \ldots, u_{n-1}u_n$ will be relaxed in path order, and therefore $D[v] \leq D'[v]$, after $n/3 + 1 + \sqrt{2cn\log n}$ iterations with probabilty at least $1 - 1/n^c$. Combining these probabilities for the distances to individual vertices into a single probability for the whole graph, after $n/3 + 1 + \sqrt{2cn\log n}$ iterations $D[v] \leq D'[v]$ for all vertices $v$ with probability at least $1 - 1/n^{c-1}$.
\end{proof}

\begin{theorem}\label{thm:cycle-after}
If the input graph $G$ has a negative cycle reachable from the source, then this can be detected as a cycle in $G_p$ after $n/3 + 2 + \sqrt{2cn\log n}$ iterations with probability at least $1 - 1/n^{c-1}$.
\end{theorem}
\begin{proof}
After $n/3 + 1 + \sqrt{2cn\log n}$ we have $D[v] \leq D'[v]$ for all vertices $v$ with probability at least $1-1/n^{c-1}$. The algorithm cannot terminate when negative cycles exist, so a relaxation must happen on the next iteration, which will cause $D[u] < D'[u]$ for some vertex $u$.
\end{proof}

In light of Theorem~\ref{thm:cycle-after} to detect negative cycles we modify Algorithm~\ref{alg:rand} by performing a cycle detection step in $G_p$ after every iteration beyond $n/3 + 2 + \sqrt{2cn\log n}$. Since $G_p$ has only one outgoing edge per vertex, cycles in it may be detected in time $O(n)$. With probability at least $1-1/n^{c-1}$ we will only perform one round of cycle detection, and in the worst case Yen's analysis guarantees that a cycle will be found after at most $n/2$ iterations. Therefore, this version of the algorithm has similar high probability time performance to our analysis for sparse graphs that do not have negative cycles.

\section{Conclusion}

We have shown that  randomizing the vertices in a graph before applying Yen's improvement of the Bellman--Ford algorithm causes the algorithm to use $2/3$ of the number of relaxations  (either in expectation or with high probability) compared to its performance in the worst case without this optimization. This is the first constant factor improvement in this basic graph algorithm since Yen's original improvements in the early 1970s. Further we can expect practical improvements in runtime inline with Yen's observations~\cite{Yen-Book}, as we have only added a single linear time step for randomization.

Our improvement for negative cycle detection works only for our sparse graph analysis. For dense graphs, we get the same bound on the number of iterations until a negative cycle can be detected with high probability using subtree traversal, but (if a negative cycle exists) we may not be able to control the number of relaxation steps per iteration of the algorithm, leading to a worse bound on the total number of relaxations than in the case when a negative cycle does not exist. However, our high probability bounds also allow us to turn the dense graph shortest path algorithm into a Monte Carlo algorithm for negative cycle detection. We simply run the algorithm for dense graphs without negative cycles, and if the algorithm runs for more than the $n^3/6+o(n^3)$ relaxations given by our high probability bound, we declare that the graph has a negative cycle, with only a small probability of an erroneous result. We leave as an open question the possibility of obtaining an equally fast Las Vegas algorithm for this case.

\section*{Acknowledgments}
This research was supported in part by the National Science
Foundation under grant 0830403, and by the
Office of Naval Research under MURI grant N00014-08-1-1015.

\raggedright
\bibliographystyle{abuser}
\bibliography{bellford}

\end{document}